\documentclass[submission,copyright,creativecommons]{eptcs}
\providecommand{\event}{CL\&C'12}
\usepackage{amssymb,amsmath,amsthm,proof,color}
\usepackage[dvipdfm]{graphicx}
\newtheoremstyle{defstyle}{\topsep}{\topsep}{\upshape}{}{\bfseries}{}{.5em}{}
\newtheoremstyle{newstyle}{\topsep}{\topsep}{}{}{\bfseries}{}{.5em}{}
\theoremstyle{newstyle}
\newtheorem{lemma}{Lemma}[section]

\newtheorem{corollary}[lemma]{Corollary}
\newtheorem{theorem}[lemma]{Theorem}
\theoremstyle{defstyle}
\newtheorem{definition}[lemma]{Definition}

\def\lm{{\Lambda\mu}}

\def\olam{{\overline{\lambda}}}
\def\a{{\alpha}}
\def\b{{\beta}}

\def\K{{\sf K}}
\def\S{{\sf S}}
\def\C{{\sf C}}
\def\W{{\sf W}}

\def\k{{\sf k}}
\def\s{{\sf s}}
\def\c{{\sf c}}
\def\w{{\sf w}}

\def\St{{\cal S}}

\def\Vars{{\rm Var}}
\def\Term{{\rm Term}}
\def\Stream{{\rm Stream}}

\def\tup#1{{({#1})}}

\def\CL{{\sf CL}}
\def\CLm{{\sf SCL}}

\def\intp#1{{[\![{#1}]\!]}}
\def\intprt#1{{[\![{#1}]\!]_{\rho,\theta}}}
\def\intP#1{{\langle\!|{#1}|\!\rangle}}

\begin{document}

\title{Extensional Models of Untyped Lambda-mu Calculus}

\author{Koji Nakazawa
\institute{Graduate School of Informatics, Kyoto University}
\email{knak@kuis.kyoto-u.ac.jp}
\and
Shin-ya Katsumata
\institute{Research Institute for Mathematical Sciences, Kyoto University}
\email{sinya@kurims.kyoto-u.ac.jp}
}

\def\titlerunning{Extensional Models of Untyped Lambda-mu Calculus}
\def\authorrunning{K.~Nakazawa and S.~Katsumata}

\maketitle

\begin{abstract}
This paper proposes new mathematical models of the untyped Lambda-mu
calculus. One is called the stream model, which is an extension of the
lambda model, in which each term is interpreted as a function from
 streams to individual data. 
%and we can construct $D_\infty$-like stream models for lambda-mu
%calculus.
The other is called the stream combinatory algebra, which is an
extension of the combinatory algebra, and it is proved that the
extensional equality of the Lambda-mu calculus is equivalent to
equality in stream combinatory algebras. In order to define the stream
combinatory algebra, we introduce a combinatory calculus $\CLm$, which
is an abstraction-free system corresponding to the
Lambda-mu calculus.  Moreover, it is shown that stream models
are algebraically characterized as a particular class of stream
combinatory algebras.
%\keywords{lambda-mu calculus, combinatory logic,
%combinatory algebra, lambda models.}
\end{abstract}

\section{Introduction}
The $\lambda\mu$-calculus was originally proposed by Parigot in
\cite{Parigot1992} as a term assignment system for the classical natural
deduction, and some variants of $\lambda\mu$-calculus have been widely
studied as typed calculi with control operators. Parigot noted that the
$\mu$-abstraction of the $\lambda\mu$-calculus can be seen as a
potentially-infinite sequence of the $\lambda$-abstraction, and
Saurin showed that an extension of the untyped
$\lambda\mu$-calculus, which was originally considered by de Groote in
\cite{deGroote1994} and was called $\Lambda\mu$-calculus by Saurin, can
be seen as a stream calculus which enjoys some fundamental properties
\cite{Saurin2005,Saurin2010a,Saurin2010b}. In particular, Saurin proved
the separation theorem of the $\Lambda\mu$-calculus in
\cite{Saurin2005}, while it does not hold in the original
$\lambda\mu$-calculus \cite{David-Py2001}.

In \cite{Saurin2010b}, Saurin also proposed the B\"{o}hm-tree
representation of the $\Lambda\mu$-terms.
% \COMMENT{In the following, \CHANGE{coloring} is omitted for the change
% $\lambda\mu$ to $\Lambda\mu$.}
That suggests a relationship between the syntax and the semantics for
the untyped $\Lambda\mu$-calculus like the neat correspondence between
the B\"{o}hm-trees and Scott's $D_\infty$ model for the untyped
$\lambda$-calculus. However, models of the untyped 
$\Lambda\mu$-calculus have not been sufficiently studied yet, so
we investigate how we can extend the results on the models of the
$\lambda$-calculus to the $\Lambda\mu$-calculus.

% For the $\lambda$-calculus, it is known that extensional
% $\lambda$-models can be seen as algebraic systems called extensional
% combinatory algebras. This paper extends this relationship to models of
% the untyped $\lambda\mu$-calculus.

In this paper, we give simple extensions of the $\lambda$-models and the
combinatory algebras, and show that they can be seen as models of the
untyped $\Lambda\mu$-calculus. First, we introduce {\em stream models}
of the untyped $\Lambda\mu$-calculus, which are extended from the
$\lambda$-models. The definition of stream model is based on the
idea that the $\Lambda\mu$-calculus represents functions on streams,
that is, in stream models, every $\Lambda\mu$-term is interpreted as
a function from streams to individual data. Then, we give a new
combinatory calculus $\CLm$, which is an extension of the ordinary
combinatory logic $\CL$, and corresponds to the $\Lambda\mu$-calculus.
The structure of $\CLm$ induces another model of the untyped
$\Lambda\mu$-calculus, called {\em stream combinatory algebra}.  We will
show that the extensional equality of the $\Lambda\mu$-calculus is
equivalent to equality in extensional stream combinatory algebras.  We
also show that the stream models are algebraically characterized as a
particular class of the stream combinatory algebras.

\section{Untyped $\Lambda\mu$-Calculus}

First, we remind the untyped $\Lambda\mu$-calculus.  We are following
the notation of \cite{Saurin2005}, because it is suitable to see the
$\lm$-calculus as a calculus operating streams.

\begin{definition}[$\lm$-calculus]\rm
% \COMMENT{We've abandoned the ``French'' notation. \CHANGE{coloring} is
%  omitted for this change.}
Suppose that there are two disjoint
sets of variables: one is the set $\Vars_T$ of {\it term variables}, denoted
by $x,y,\cdots$, and the other is the set $\Vars_S$ of {\it stream variables},
denoted by $\a,\b,\cdots$. Terms and axioms of the $\lm$-calculus
are given in the Fig. \ref{fig:lm-calculus}.
\begin{figure*}[t]
Terms:
\[
 M,N::=x \mid \lambda x.M \mid MN \mid \mu\a.M \mid M\a
\]

Axioms:
\begin{align*}
 (\lambda x.M)N & =_{\beta_T} M[x:=N] \\
 (\mu\a.M)\b & =_{\beta_S} M[\a:=\b] \\
 \lambda x.Mx & =_{\eta_T} M & \mbox{($x\not\in FV(M)$)}\\
 \mu\a.M\a & =_{\eta_S} M & \mbox{($\a\not\in FV(M)$)}\\
 (\mu\a.M)N & =_{\mu} \mu\a.M[P\a:=PN\a] 
\end{align*}
\caption{Untyped $\lm$-calculus}
\label{fig:lm-calculus}
\end{figure*}
The set of the $\lm$-terms is denoted by $\Term_{\lm}$. We
use the following abbreviations: $\lambda x_1 x_2\cdots x_n.M$ denotes
$\lambda x_1.(\lambda x_2.(\cdots(\lambda x_n.M)\cdots))$ and similarly
for $\mu$, $MA_1\cdots A_n$ denotes $(\cdots(MA_1)\cdots)A_n$,
in which each $A_i$ denotes either a term or a stream variable, and the
top-level parentheses are also often omitted. Variable occurrences of
$x$ and $\a$ are bound in $\lambda x.M$ and $\mu\a.M$,
respectively. Variable occurrences which are not bound are called free,
and $FV(M)$ denotes the set of variables freely occurring in $M$.
In the axioms, $M[x:=N]$ and $M[\a:=\b]$ are the
usual capture-avoiding
substitutions, and $M[P\a:=PN\a]$ recursively replaces each
subterm of the form $P\a$ in $M$ by $PN\a$. The relation
$M=_{\lm}N$ is the compatible equivalence relation defined from
the axioms.

{\em Contexts} are defined as $K ::= []\a \mid K[[]M]$, and
$K[M]$ is defined in a usual way. The substitution $M[P\a:=K[P]]$
recursively replaces each subterm of the form $P\a$ in $M$ by $K[P]$.
\end{definition}

Each context has the form $[]M_1\cdots M_n\a$ and it corresponds to
a stream data, the initial segment of which is $M_1 \cdots M_n$ and the
rest is $\a$. It is easy to see that
$K[\mu\a.M]=_{\lm}M[P\a:=K[P]]$ for any term $M$ and any
context $K$.

The untyped $\lm$-calculus can be seen as a calculus operating
streams, in which the $\mu$-abstractions represent functions on streams,
and a term $MN_0\cdots N_n\alpha$ means a function application of $M$
to the stream data $[]N_0\cdots N_n\alpha$. For example, the term
${\sf hd} = \lambda x.\mu\a.x$ is the function to get the head element
of streams since we have ${\sf hd}\,N_0\cdots N_n\beta =_{\beta_T}
(\mu\a.N_0)N_1\cdots N_n\b =_{\mu} (\mu\a.N_0)\b =_{\beta_S} N_0$. For
another example, we have a term ${\sf nth}$ representing the function
which takes a stream and a numeral $c_n$ and returns the $n$-th element
of the stream. The term ${\sf nth}$ is defined as
\[
 Y(\lambda fx.\mu\a.\lambda y.{\sf if}\ ({\sf zero?}\ y)\ {\sf
 then}\ x\ {\sf else}\ f\a(y-1)),
\]
where $Y$ is a fixed point operator in the $\lambda$-calculus, and we
have
\[
 {\sf nth}\,N_0 N_1 N_2 \cdots N_n \beta\,c_i =_{\lm} N_i
\]
for any $0\le i\le n$. However, the $\lm$-calculus has no term
representing a stream, and that means $\lm$-terms do not directly
represent any function which returns streams.

In Parigot's original $\lambda\mu$-calculus \cite{Parigot1992}, terms of
the form $P\a$, which are originally denoted by $[\a]P$, are
distinguished as {\it named terms} from the ordinary terms, and bodies
of $\mu$-abstractions are restricted to the named terms. On the other
hand, we consider $P\a$ as an ordinary term and any term can be the body
of $\mu$-abstraction in the $\lm$-calculus. For example,
neither $M\a N$ nor $\mu \a.x$ is allowed as a term in the original
$\lambda\mu$-calculus, whereas they are well-formed terms in the
$\lm$-calculus. Such extensions of the $\lambda\mu$-calculus in
which the named terms are not distinguished have been originally studied
by de Groote \cite{deGroote1994}, and Saurin \cite{Saurin2005} considered a
reduction system with the $\eta$-reduction,
%  This calculus is called
% $\Lambda\mu$-calculus by Saurin in \cite{Saurin2005},
where another axiom
\[
 \mu\a.M \to_{\it fst} \lambda x.\mu\a.M[P\a:=Px\a]
\]
is chosen instead of ($\mu$). For extensional
equational systems, the axioms ($\mu$) and ({\it fst}) are equivalent
since
\[
 \mu\a.M
 =_{\eta_T} \lambda x.(\mu\a.M)x
 =_{\mu} \lambda x.\mu\a.M[P\a:=Px\a], \mbox{ and}
\]
\[
 (\mu\a.M)N 
 =_{\it fst} (\lambda x.\mu\a.M[P\a:=Px\a])N
 =_{\it \beta_T} \mu\a.M[P\a:=PN\a]).
\]
% -*- mode: LaTeX; TeX-master: "main.tex"; -*-
\newcommand{\Bf}[1]{{\bf #1}}
\newcommand{\CC}{\Bf C}
\newcommand{\Set}{{\rm Set}}
\newcommand{\Terms}{{\rm Term}}
\section{Stream Models}\label{sec:smodel}

In this section, we introduce extensional stream models for the
untyped $\Lambda\mu$-calculus. The definition follows the idea that
the $\Lambda\mu$-terms represent functions on streams.

% First, we give a definition of the extensional stream models,
%  straightforwardly following the idea that the $\lambda\mu$-calculus is
%  a stream calculus. However, it depends on the definability of the
%  interpretation of $\lambda\mu$-terms, so we give an algebraic
%  characterization of the extensional stream models induced by the
%  structure of $\CLm$, which is independent of the syntax of the
%  $\lambda\mu$-calculus.
% and it can be seen as another definition of the
%  models
% We consider only stream models which we may have to call extensional
% stream models since they are always extensional $\lambda$-models, and
% the $\eta$-axioms hold in them. However, in the following, they are just
% called stream models for simplicity.

\subsection{Definition of Extensional Stream Models}

%% \COMMENT{It is pointed out that we should call $S$ the set of
%%   streams (and call each element in $S$ a stream. How about call
%%   $\tup{S,::}$ ``stream set''?}
In the following, we use $\olam$ to represent meta-level functions. A
{\em stream set} over a set $D$ is a pair $\tup{S,::}$ of a set
$S$ and a bijection $(::):D\times S\rightarrow S$.  A typical
stream set over $D$ is the $\Bf N$-fold product of $D$, that
is, $\tup{D^{\Bf N},::}$ where
\begin{align*}
 d::s & = \olam n\in\Bf N.\begin{cases}
	      d & \mbox{($n=0$)}\\
	      s(n-1) & \mbox{($n>0$)}.
	     \end{cases}
\end{align*}
% We define the head and tail operations for the stream by ${\it
%   hd}=\pi_1\circ(::)^{-1}:S\to D$ and ${\it tl}=\pi_2\circ(::)^{-1}:S\to S$,
% respectively.
For a function $f:D\times S\to E$, $\olam d::s\in S.f(d,s)$ denotes
the function $f\circ(::)^{-1}:S\rightarrow E$.

\begin{definition}[Extensional stream models]\label{def:ext}\rm
An {\it extensional stream model}
is a tuple $\tup{D,S,[S\to D],::,\Psi}$ such that

1. $\tup{S,::}$ is a stream set over $D$.

2. $[S\to D]$ is a subset of $S\to D$.

3. $\Psi:[S\to D]\to D$ is a bijection. We write its inverse by $\Phi$.

4. There is a (necessarily unique) function
$\intp{-}:\Terms_{\lm}\times (\Vars_T\to D)\times(\Vars_S\to
S)\to D$, called {\em meaning function}, such that
       \begin{align*}
	\intp{x}_{\rho,\theta} & = \rho(x)\\
	\intp{\lambda x.M}_{\rho,\theta} & = \Psi(\olam d::s\in
	S.\Phi(\intp{M}_{\rho[x\mapsto d],\theta})(s))\\
	\intp{MN}_{\rho,\theta} & = \Psi(\olam s\in
	S.\Phi(\intp{M}_{\rho,\theta})(\intp{N}_{\rho,\theta}::s))\\
	\intp{\mu\a.M}_{\rho,\theta} & = \Psi(\olam s\in
	S.\intp{M}_{\rho,\theta[\a\mapsto s]})\\
	\intp{M\a}_{\rho,\theta} & = \Phi(\intp{M}_{\rho,\theta})(\theta(\a)).
       \end{align*}
Here $\rho[x\mapsto d]$ is defined by
\[
 \rho[x\mapsto d](y) = \begin{cases}
			d & \mbox{($x=y$)}\\
			\rho(y) & \mbox{($x\not=y$)},
		       \end{cases}
\]
and $\theta[\a\mapsto s]$ is defined similarly.
We use the notation $d\star s$ to denote $\Phi(d)(s)$ for $d\in D$ and
 $s\in S$.
\end{definition}
The condition 4 requires that each argument of $\Psi$ is contained in
$[S\to D]$.
In the next subsection, we show that extensional stream models can be
obtained from the solutions of the simultaneous recursive equations $D
\times S \cong S$ and $S \Rightarrow D \cong D$ in a well-pointed CCC
(Theorem \ref{thm:cat-lambda-model}).
%\COMMENT{A nontrivial example of extensional stream model will be
%  given in the next subsection.}

\begin{lemma}\label{lem:smodel}
The following hold.

1. $\intprt{M[x:=N]}
 =\intp{M}_{\rho[x\mapsto \intprt{N}],\theta}$.

2. $\intprt{M[\a:=\b]}
 =\intp{M}_{\rho,\theta[\a\mapsto\theta(\b)]}$.

3. $\intprt{M[P\a:=PN\a]}
 =\intp{M}_{\rho,\theta[\a\mapsto\intprt{N}{}::\theta(\a)]}$.
\end{lemma}

\begin{proof}
 By induction on $M$. We show only the case of $M=M'\a$ for 3.
 \begin{align*}
  &\intprt{(M'\a)[P\a:=PN\a]}\\
  =& \intprt{M'[P\a:=PN\a]N\a}\\
  =& \intprt{M'[P\a:=PN\a]}\star(\intprt{N}::\theta(\a))\\
  =& \intp{M'}_{\rho,\theta[\a\mapsto\intprt{N}::\theta(\a)]}
  \star(\intprt{N}::\theta(\a)) & \mbox{(by IH)}\\
  =& \intp{M'\a}_{\rho,\theta[\a\mapsto\intprt{N}::\theta(\a)]}.
 \end{align*}
\end{proof}

\begin{theorem}[Soundness]
 Let $D$ be an arbitrary extensional stream model.  If
$M=_{\lm}N$, then $\intp{M}_{\rho,\theta}=\intp{N}_{\rho,\theta}$
holds in $D$ for any $\rho$ and $\theta$.
\end{theorem} 

\begin{proof}
By induction on $M=_{\lm}N$.  We show only two cases, and the
other cases are similarly proved by Lemma \ref{lem:smodel}.

Case ($\beta_T$).
\begin{align*}
 \intprt{(\lambda x.M)N}
 & = \Psi(\olam s.( \Psi(\olam d'::s'.(\intp{M}_{\rho[x\mapsto
 d'],\theta})\star s') )\star(\intprt{N}::s))\\
% & = \Psi(\olam s.(\olam d'::s'.(\intp{M}_{\rho[x\mapsto
% d'],\theta})\star s') (\intprt{N}::s))\\ 
 & = \Psi(\olam s.(\intp{M}_{\rho[x\mapsto\intprt{N}],\theta})\star s)\\ 
% & = \Psi(\Phi(\intp{M}_{\rho[x\mapsto\intprt{N}],\theta}))\\ 
 & = \intp{M}_{\rho[x\mapsto \intprt{N}],\theta}\\ 
 & = \intprt{M[x:=N]} \qquad \mbox{(by Lemma \ref{lem:smodel}.1)}
\end{align*}

% Case ($\beta_S$).
% \begin{align*}
%  \intprt{(\mu\a.M)\b}
%  &= (\Psi(\olam s.\intp{M}_{\rho,\theta[\a\mapsto s]} ))\star\theta(\b)\\
%  &= \intp{M}_{\rho,\theta[\a\mapsto \theta(\b)]}\\
%  & = \intprt{M[\a:=\b]}. \qquad \mbox{(by Lemma \ref{lem:smodel}.2)}
% \end{align*}
%
% Case ($\eta_T$).
% \begin{align*}
%  \intprt{\lambda x.(M)x}
% % &= \Psi(\olam d::s.\intp{Mx}_{\rho[x\mapsto d],\theta}\star s)\\
%  &= \Psi(\olam d::s.(\Psi(\olam s'.\intp{M}_{\rho[x\mapsto d],\theta}\star(d::s')))\star s)\\
%  &= \Psi(\olam d::s.\intp{M}_{\rho[x\mapsto d],\theta}\star(d::s))\\
%  &= \Psi(\olam d::s.\intprt{M}\star(d::s))
%  \qquad \mbox{(by $x\not\in FV(M)$)}\\
% % &= \Psi(\Phi(\intprt{M}))\\
%  &= \intprt{M}
% \end{align*}
%
% Case ($\eta_S$).
% \begin{align*}
%  \intprt{\mu\a.(M)\a}
%  &= \Psi(\olam s.\intp{M}_{\rho,\theta[\a\mapsto s]}\star s)\\
%  &= \Psi(\olam s.\intprt{M}\star s) \qquad \mbox{(by $\a\not\in FV(M)$)}\\
% % &= \Psi(\Phi(\intprt{M}))\\
%  &= \intprt{M}.
% \end{align*}

Case ($\mu$).
\begin{align*}
 \intprt{(\mu\a.M)N}
 &= \Psi(\olam s.(\Psi(\olam s'.\intp{M}_{\rho,\theta[\a\mapsto s']}))\star(\intprt{N}::s))\\
 &= \Psi(\olam s.\intp{M}_{\rho,\theta[\a\mapsto \intprt{N}::s]})
\end{align*}
On the other hand, if we let $\theta'=\theta[\a\mapsto s]$, then the
following holds.
\begin{align*}
 \intprt{\mu\a.M[P\a:=PN\a]}
% &= \Psi(\olam s.\intp{M[P\a:=PN\a]}_{\rho,\theta'})\\
 &= \Psi(\olam s.\intp{M}_{\rho,\theta'[\a\mapsto
 \intp{N}_{\rho,\theta'}::s]}) \qquad \mbox{(by Lemma \ref{lem:smodel}.3)}\\ 
 &= \Psi(\olam s.\intp{M}_{\rho,\theta[\a\mapsto
 \intp{N}_{\rho,\theta}::s]}) \qquad \mbox{(by $\a\not\in FV(N)$)}
\end{align*}
\end{proof}

\begin{theorem}\label{thm:lambda-model}
 Every extensional stream model is an extensional $\lambda$-model in
 which the interpretation of $\lambda$-terms coincides with the
 interpretation in the stream model.
\end{theorem}

\begin{proof}
 Let $D$ be an extensional stream model, then we can define $[D\to
   D]$, $\Phi_0:D\to[D\to D]$, and $\Psi_0:[D\to D]\to D$ as follows.
 \begin{align*}
  [D\to D] &:= \{f:D\to D \mid (\olam d::s\in
  S.(f(d))\star s)\in[S\to D] \}\\
  \Phi_0(d) &:= \olam d'\in D.\Psi(\olam s\in S.d\star(d'::s))\\
  \Psi_0(f) &:= \Psi(\olam d::s\in S.(f(d))\star s)
 \end{align*}
 Note that these are variants of {\sf eval} and {\sf abst} in
 \cite{Streicher-Reus1998}, and just based on the isomorphism $D\times
 S\simeq S $. Then, it is easily checked that $D$ is a $\lambda$-model
 with $\Phi_0$ and $\Psi_0$. The interpretation of the $\lambda$-terms
 in the $\lambda$-model, denoted $\intp{\cdot}^\lambda$ here,
 coincides with the interpretation in the stream model as follows:
 \begin{align*}
  \intp{\lambda x.M}^\lambda_\rho & = \Psi_0(\olam d\in
  D.\intp{M}^\lambda_{\rho[x\mapsto d]})\\
  & = \Psi(\olam d'::s'\in S.(\intp{M}^\lambda_{\rho[x\mapsto
  d']})\star s') & \mbox{(by Def. of $\Psi_0$)}\\
  & = \intp{\lambda x.M}_\rho & \mbox{(by IH)},\\
  \intp{MN}^\lambda_\rho
  & = \Phi_0(\intp{M}^\lambda_\rho)(\intp{N}^\lambda_\rho)\\
  & = \Psi(\olam s\in S.(\intp{M}^\lambda_\rho)\star(\intp{N}^\lambda_\rho::s)) & \mbox{(by Def. of $\Phi_0$)}\\
  & = \intp{MN}_\rho & \mbox{(by IH)}.
 \end{align*}
\end{proof}

\subsection{Categorical Stream Models}\label{sec:categorical}

In a categorical setting, a solution $(D,S)$ of the following
simultaneous recursive  equations in a CCC provides a model of the
$\lm$-calculus.
\begin{equation}
  \label{eq:streameq}
  D\times S\simeq S,\quad S\Rightarrow D\simeq D  
\end{equation}
\begin{definition}[Categorical stream models]\rm
  A {\em categorical stream model} in a CCC $\CC$ is a tuple
  $\tup{D,S,c,\psi}$ of objects $D$ and $S$, and isomorphisms
  $c:D\times S\to S$ and $\psi:S\Rightarrow D\to D$.
\end{definition}
When $\CC$ has countable products, the solutions of the following
recursive equation:
\begin{equation}
  \label{eq:nakazawa}
  D^{\Bf N}\Rightarrow D\simeq D
\end{equation}
yield categorical stream models, as we always have $D^{\Bf N}\simeq
D\times D^{\Bf N}$.

Given a categorical stream model $\tup{D,S,c,\psi}$, we can interpret
$\lm$-terms as a morphism $\intp M_{\vec x,\vec\alpha}:D^{|\vec
  x|}\times S^{|\vec\alpha|}\to D$, where $\vec x$
(resp. $\vec\alpha$) is a finite sequence of distinct term (stream)
variables such that every free term (stream) variable in $M$ occurs in
$\vec x$ ($\vec\alpha$), and $|\vec x|$ ($|\vec\alpha|$) is the length
of $\vec x$ ($\vec\alpha$). We omit the details of this
interpretation, as it is a straightforward categorical formulation of
the meaning function in Definition \ref{def:ext}.

When the underlying CCC $\CC$ of a categorical stream model is
well-pointed (that is, the global element functor
$\CC(1,-):\CC\to\Set$ is faithful), we can convert it to an
extensional stream model.
\begin{theorem}\label{thm:cat-lambda-model}
  Let $\CC$ be a well-pointed CCC.  For any categorical stream model
  $(D,S,c,\psi)$ in $\CC$, the following tuple is an extensional
  stream model:
  \begin{displaymath}
    \tup{\CC(1,D),~\CC(1,S),~\{\CC(1,f)~|~f\in\CC(S,D)\},~
    \olam (f,g).c\circ\langle f,g\rangle,~\Psi},
  \end{displaymath}
  where $\Psi$ is the function defined by $\Psi(\CC(1,f))=\psi\circ\lambda(f\circ\pi_2)$.
\end{theorem}
For instance, in the well-pointed CCC of pointed CPOs and all
continuous functions, the standard inverse limit method
\cite{Scott1972,Smyth-Plotkin1982} applied to the
following embedding-projection pair $( e : D_0 \rightarrow D_0^{\bf N}
\Rightarrow D_0, p : D_0^{\bf N} \Rightarrow D_0 \rightarrow D_0 )$:
\[ e \left( x \right) = \overline{\lambda} y \in D_0^{\bf N} .x,
\hspace{1em} p \left( f \right) = f \left( \bot, \ldots \right) \] on
a pointed CPO $D_0$ containing at least two elements yields a
non-trivial solution of \eqref{eq:nakazawa}. From this solution, an
extensional stream model is derived by Theorem
\ref{thm:cat-lambda-model}.  This model distinguishes $\intp{ \lambda
  x y.x}$ and $\intp{\lambda x y.y }$, hence, we obtain a model
theoretic consistency proof of the $\lm$-calculus (consistency also
follows from confluence, which has been proved in
\cite{Saurin2010c}).

\section{Stream Combinatory Algebra}

% \COMMENT{It is pointed out that $\CLm$ is not a combinatory calculus in
% a strict sense, but a term rewriting.}
We give another model of the untyped $\lm$-calculus. It is called
stream combinatory algebra, which is an extension of the combinatory
algebra corresponding to the combinatory logic $\CL$.

\subsection{Combinatory Calculus $\CLm$}

We introduce a new combinatory calculus $\CLm$, and show that $\CLm$ is
equivalent to the $\lm$-calculus.  This result is an extension of
the equivalence between the $\lambda$-calculus and the untyped variant
of the ordinary combinatory logic $\CL$ with the combinators $\K$ and
$\S$.
% In fact, $\CLm$ is an
% extension of $\CL$ by an additional operator for streams, extra five
% constants for streams, and extensionality rules.
In $\CLm$, the combinators $\K$ and $\S$ are denoted by $\K_0$ and
$\S_0$, respectively.

\begin{definition}[$\CLm$]\rm
Similarly to the $\lm$-calculus, $\CLm$ has two sorts of
variables: term variables $\Vars_T$ and stream variables
$\Vars_S$. Constants, terms, streams, axioms, and extensionality rules
of $\CLm$ are given in Fig. \ref{fig:CLm}.
\begin{figure*}[t]
Constants:
\[
 C ::= \K_0 \mid \K_1 \mid \S_0 \mid \S_1 \mid \C_{10} \mid \C_{11} \mid \W_1
\]

Terms:
\[
 T,U::= C \mid x \mid T\cdot U \mid T\star \St
\]

Streams:
\[
 \St ::= \alpha \mid T::\St
\]

Axioms:
\begin{align*}
 \K_0 \cdot T_1 \cdot T_2 & = T_1
 & \K_1 \cdot T_1 \star \St_2 & = T_1\\
 \S_0 \cdot T_1 \cdot T_2 \cdot T_3 & = T_1 \cdot T_3 \cdot (T_2 \cdot T_3)
 & \S_1 \cdot T_1 \cdot T_2 \star \St_3 & = T_1 \star \St_3 \cdot (T_2
 \star \St_3)\\
 \C_{10} \cdot T_1 \star \St_2 \cdot T_3 & = T_1 \cdot T_3 \star \St_2
 & \C_{11} \cdot T_1 \star \St_2 \star \St_3 & = T_1 \star \St_3 \star \St_2\\
 \W_1 \cdot T_1 \star \St_2 & = T_1 \star \St_2 \star \St_2
 & T_1\star(T_2::\St_3) & = T_1\cdot T_2\star\St_3
\end{align*}

Extensionality rules:
\[
 \infer[\mbox{($\zeta_T$)}]{T=U}{T\cdot x=U\cdot x & x\not\in FV(T)\cup FV(U)}
 \qquad
 \infer[\mbox{($\zeta_S$)}]{T=U}{T\star\a=U\star\a & \a\not\in FV(T)\cup FV(U)}
\]
\caption{$\CLm$}
\label{fig:CLm}
\end{figure*}
 The set of the $\CLm$-terms and the set of the $\CLm$-streams are
 denoted by $\Term_{\CLm}$ and $\Stream_{\CLm}$, respectively.  The set
 of variables occurring in $T$ is denoted by $FV(T)$. We suppose that
 the binary function symbols $(\cdot)$ and $(\star)$ have the same
 associative strength, and both are left associative. For example,
 $T_1\cdot T_2\star \St_3 \cdot T_4$ denotes $((T_1\cdot
 T_2)\star\St_3)\cdot T_4$. The substitutions $T[x:=T']$ and
 $T[\a:=\St]$ are defined straightforwardly.
The relation $T=_\CLm U$ is the compatible equivalence relation defined
from the axioms and the extensionality rules.
% Terms containing no $::$-operator are called {\em cons normal form}
%  ({\em cons-nf}), and we say that $T$ has a cons-nf $U$ if $U$ is a
%  cons-nf and we have $T=_\CLm U$ only with the last axiom
%  $T_1\star(T_2::\St_3)=T_1\cdot T_2\star\St_3$.
\end{definition}

%Intuitively, stream variables and stream terms denote streams.
% As for the combinatory logic $\CL$, the operation $\cdot$ corresponds to
% and the
The new operation $(\star)$ represents the function application for
streams, which corresponds to the application $M\a$ in the
$\lm$-calculus.

% The cons-nf's are in the form of $T_1\cdot T_2\cdot \cdots \star\a$, and
% the following is easily proved.

% \begin{lemma}
%  Every $\CLm$-term has a unique cons-nf.
% \end{lemma}

% \begin{proof}
%  The reduction relation induced by the rule $T_1\star(T_2::\St_3) \to
%  T_1\cdot T_2\star\St_3$ is strongly normalizing and confluent, and we
%  can show that if $T$ contains $::$ then there is $U$ such that $T\to
%  U$.
% \end{proof}

In the following, we think that the term of the form $T_1\cdot
T_2\star\St_3$ is simpler than $T_1\star(T_2::\St_3)$, and that is
formalized as the following measure $|T|$.

\begin{definition}\rm
The measure $|T|$ of $\CLm$-terms is defined as $|T|={\sf c}(T) + {\sf
m}(T)$, where ${\sf c}(T)$ is the number of the symbol $::$ occurring in
$T$, and ${\sf m}(T)$ is the number of nodes of the syntax tree of $T$.
% \begin{align*}
%  {\sf m}(C) & = 1 & {\sf m}(\a) & = 1\\
%  {\sf m}(x) & = 1 & {\sf m}(T::\St) & = {\sf m}(T) + {\sf m}(\St)\\
%  {\sf m}(T\cdot U) & = {\sf m}(T)+{\sf m}(U)\\
%  {\sf m}(T\star \St) & = {\sf m}(T)+{\sf m}(\St)
% \end{align*}
\end{definition}

It is easily seen that if $T$ is a subterm of $U$ then $|T|<|U|$,
and $|T_1\cdot T_2\star\St_3|<|T_1 \star(T_2::\St_3)|$, which follows
from ${\sf m}(T_1\cdot T_2\star\St_3)={\sf m}(T_1\star(T_2::\St_3))$.

% \begin{lemma}\label{lem:scl-measure}
%  The following hold.
%
% 1. If $T$ is a subterm of $U$, then $|T|<|U|$
%
% 2. $|T_1\cdot T_2\star\St_3|<|T_1 \star(T_2::\St_3)|$.
% \end{lemma}
%
% \begin{proof}
% 1 is proved by induction on $T$, and 2 is proved by ${\sf m}(T_1\star(T_2::\St_3))={\sf m}(T_1\cdot T_2\star\St_3)$.
% \end{proof}

The $\lm$-calculus and $\CLm$ are equivalent through the
following translations.

\begin{definition}[Translations between $\lm$ and $\CLm$]\rm
 1. For $T\in\Term_{\CLm}$ and $x\in\Vars_T$, we define the
$\CLm$-term $\lambda^* x.T$ inductively on $|T|$ as follows:
\begin{align*}
 \lambda^*x.x & = \S_0\cdot\K_0\cdot\K_0\\
 \lambda^*x.T & = \K_0\cdot T & \mbox{($x\not\in FV(T)$)}\\
 \lambda^*x.(T\cdot U) & = \S_0\cdot (\lambda^*x.T) \cdot (\lambda^*x.U)\\
 \lambda^*x.(T\star\a) & = \C_{10}\cdot (\lambda^*x.T) \star\a\\
 \lambda^*x.(T\star(U::\a)) & = \lambda^*x.(T\cdot U\star\a).
\end{align*}
% For an arbitrary $\CLm$-term $T$, $\lambda^*x.T$ is defined as
%  $\lambda^*x.U$ where $U$ is the unique cons-nf of $T$.

For $T\in\Term_{\CLm}$ and $\a\in\Vars_S$, we define
the $\CLm$-term $\mu^*\a.T$ inductively on $|T|$ as follows:
\begin{align*}
 \mu^*\a.T & = \K_1 \cdot T & \mbox{($\a\not\in FV(T)$)}\\
 \mu^*\a.(T\cdot U) & = \S_1 \cdot (\mu^*\a.T) \cdot (\mu^*\a.U)\\
 \mu^*\a.(T\star\a) & = \W_1 \cdot (\mu^*\a.T)\\
 \mu^*\a.(T\star\b) & = \C_{11} \cdot (\mu^*\a.T) \star \b &
 \mbox{($\a\not=\b$)}\\
 \mu^*\a.(T\star(U::\a)) & = \mu^*\a.(T\cdot U\star\a).
\end{align*}
% For an arbitrary $\CLm$-term $T$, $\mu^*\a.T$ is defined as
% $\mu^*\a.U$ where $U$ is the unique cons-nf of $T$.

Then the mapping $M^*$ from $\Term_{\lm}$ to $\Term_{\CLm}$ is
defined by
\begin{align*}
 x^* & = x\\
 (\lambda x.M)^* & = \lambda^* x.M^* 
 & (MN)^* & = M^*\cdot N^*\\
 (\mu\a.M)^* & = \mu^*\a.M^*
 & (M\a)^* & = M^*\star\a.
\end{align*}

2. The mappings $T_*$ from $\Term_{\CLm}$ to ${\rm Term_{\lambda\mu}}$
 and $\St_*$ from $\Stream_{\CLm}$ to contexts are defined by
\begin{align*}
 (\K_0)_* & = \lambda xy.x & x_* & = x \\
 (\K_1)_* & = \lambda x.\mu\a.x & (T\cdot U)_* & = T_* U_* \\
 (\S_0)_* & = \lambda xyz.xz(yz) & (T\star \St)_* & = \St_*[T_*] \\
 (\S_1)_* & = \lambda xy.\mu\a.x\a(y\a)\\
 (\C_{10})_* & = \lambda x.\mu\a.\lambda y.xy\a & \a_* & = []\a\\
 (\C_{11})_* & = \lambda x.\mu\a\b.x\b\a & (T::\St)_* & = \St_*[[]T_*].\\
 (\W_1)_* & = \lambda x.\mu\a.x\a\a 
\end{align*}
\end{definition}

By the extensionality of $\CLm$, the definitions of $\lambda^*
x.T$ and $\mu^*\a.T$ such that 1 of the following lemma holds are unique
modulo $=_\CLm$.

\begin{lemma}\label{lem:lm-and-scl1}
The following hold.

1. $(\lambda^* x.T)\cdot U =_\CLm T[x:=U]$ and $(\mu^*\a.T)\star\St
 =_\CLm T[\a:=\St]$.

2. If $T=_\CLm U$, then $\lambda^*x.T=_\CLm \lambda^*x.U$ and
$\mu^*\a.T=_\CLm\mu^*\a.U$.
\end{lemma}

\begin{proof}
1. By induction on $|T|$.

2. By 1, we have $(\lambda^*x.T)\cdot x =_\CLm T$ and
$(\lambda^*x.U)\cdot x =_\CLm U$. Since $T =_\CLm U$, we have
 $(\lambda^*x.T)\cdot x =_\CLm (\lambda^*x.U)\cdot x$, and hence
$\lambda^*x.T=_\CLm\lambda^*x.U$ by ($\zeta_T$).
\end{proof}

\begin{lemma}\label{lem:lm-and-scl2}
 The following hold.

1. $(M[x:=N])^* =_\CLm M^*[x:=N^*]$.

2. $(M[\a:=\b])^* =_\CLm M^*[\a:=\b]$.

3. $(M[P\a:=PN\a])^* =_\CLm M^*[\a:=N^*::\a]$.
\end{lemma}

\begin{proof}
 By induction on $M$. We show only the case of $M=\lambda y.M'$ for
 1. We suppose that $y\not\in FV(N)$ and $y\not\equiv x$ by renaming
 bound variables.  We have $((\lambda y.M')[x:=N])^*\cdot y = (\lambda
 y.M'[x:=N])^*\cdot y = (\lambda^* y.(M'[x:=N])^*)\cdot y =_{\CLm}
 (M'[x:=N])^*$ by Lemma \ref{lem:lm-and-scl1}.1, and it is identical
 with $M'^*[x:=N^*]$ by the induction hypothesis.  On the other hand, we
 have $(\lambda^* y.M'^*[x:=N^*])\cdot y =_{\CLm} M'^*[x:=N^*]$. Hence,
 by ($\zeta_T$), we have $((\lambda y.M')[x:=N])^* =_{\CLm} (\lambda
 y.M')^*[x:=N^*]$.
\end{proof}

\begin{lemma}\label{lem:lm-and-scl3}
The following hold.

1. $M=_{\lm}N$ implies $M^*=_\CLm N^*$.

2. $T=_\CLm U$ implies $T_*=_{\lm} U_*$.

3. $(M^*)_* =_{\lm} M$.

4. $(T_*)^* =_{\CLm} T$ and $(\St_*[M])^* =_\CLm M^*\star \St$
\end{lemma}

\begin{proof}
By the previous lemmas, they are proved by induction straightforwardly.
\end{proof}

It is shown that the combinatory calculus $\CLm$ is equivalent to the
$\lm$-calculus in the following sense.

\begin{theorem}\label{thm:lm-and-scl}
1. For any $\lm$-terms $M$ and $N$, $M=_{\lm}N$ iff $M^*=_{\CLm}N^*$.

2. For any $\CLm$-terms $T$ and $U$, $T=_{\CLm}U$ iff $T_*=_{\lm}U_*$.
\end{theorem}

\begin{proof}
 1. The only-if part is Lemma \ref{lem:lm-and-scl3}.1, and the if part
 is proved by Lemma \ref{lem:lm-and-scl3}.2 and \ref{lem:lm-and-scl3}.3 as
$M=_{\lm} (M^*)_* =_{\lm} (N^*)_* =_{\lm} N$.

 2. Similar to 1 by 1, 2, and 4 of Lemma \ref{lem:lm-and-scl3}.
\end{proof}
\subsection{Stream Combinatory Algebra}

The stream combinatory algebras are given as models of $\CLm$. Since
$\CLm$ is equivalent to the $\lm$-calculus in the sense of
Theorem \ref{thm:lm-and-scl}, they are also models of the untyped
$\lm$-calculus.

\begin{definition}[Stream combinatory algebras]\rm
(1) For non-empty sets $D$ and $S$, a tuple $\tup{D,S,\cdot,\star,::}$ is
 called a {\it stream applicative structure} if $(\cdot):D\times D\to D$,
 $(\star):D\times S\to D$, and $(::)\,:D\times S\to S$ are mappings such
 that
\[
 d_1\star(d_2::s_3) = d_1\cdot d_2\star s_3
\]
for any $d_1,d_2\in D$ and $s_3\in S$.

(2) A stream applicative structure $D$ is {\it
extensional} if the following hold for any $d,d'\in D$:
\begin{align*}
 \forall d_0\in D[d\cdot d_0 = d'\cdot d_0]
\ \mbox{implies}\ d=d',\\
 \forall s_0\in S[d\star s_0= d'\star s_0]
\ \mbox{implies}\ d=d'.
\end{align*}

(3) A stream applicative structure $D$ is called a
 {\it stream combinatory algebra} if $D$ contains distinguished elements
 $\k_0$, $\k_1$, $\s_0$, $\s_1$, $\c_{10}$, $\c_{11}$, and $\w_1$ such
 that the following hold for any $d_1,d_2,d_3\in D$ and $s_2,s_3\in
 S$.
\begin{align*}
 \k_0 \cdot d_1 \cdot d_2 & = d_1
 & \k_1 \cdot d_1 \star s_2 & = d_1\\
 \s_0 \cdot d_1 \cdot d_2 \cdot d_3 & = d_1 \cdot d_3 \cdot (d_2 \cdot d_3)
 & \s_1 \cdot d_1 \cdot d_2 \star s_3 & = d_1 \star s_3 \cdot (d_2 \star s_3)\\
 \c_{10} \cdot d_1 \star s_2 \cdot d_3 & = d_1 \cdot d_3 \star s_2
 & \c_{11} \cdot d_1 \star s_2 \star s_3 & = d_1 \star s_3 \star s_2\\
 \w_1 \cdot d_1 \star s_2 & = d_1 \star s_2 \star s_2
\end{align*}
\end{definition}

% A stream applicative structure $\tup{D;\cdot,\star}$ can be seen as an
% infinitary algebra with a binary operator $\cdot$ and an $\omega$-ary
% operator $\bullet$, in which $\star$ is defined as
% $d\star(d_1,d_2\cdots)=\bullet(d,d_1,d_2,\cdots)$. Then, roughly speaking, the
% condition for the stream applicative structure means that $\bullet$ is the
% extension of the binary operation $\cdot$ to $\omega$-sequences, that
% is, $\bullet(d_1,d_2,d_3\cdots)=d_1\cdot d_2\cdot d_3\cdots$.

Note that, for a stream applicative structure
$\tup{D,S,\cdot,\star,::}$, the set $S$ is not necessarily a
stream set on $D$ in the sense of Section \ref{sec:smodel}, and
we will call $D$ {\em standard} if $\tup{S,::}$ is a stream set
on $D$.

It is clear that any stream combinatory algebra is always a combinatory
 algebra by ignoring the stream part, that is, $(\star)$, $(::)$, $\k_1$,
 $\s_1$, $\c_{10}$, $\c_{11}$, and $\w_1$. Therefore, any extensional
 stream combinatory algebra is an extensional combinatory algebra, and
 hence an extensional $\lambda$-model.

We can interpret $\CLm$ in stream combinatory algebras in a
straightforward way.

\begin{definition}[Interpretation of $\CLm$]\rm
 Let $\tup{D,S,\cdot,\star,::}$ be a stream combinatory algebra.  The
{\em meaning functions} $\intP{-}^T:\Term_{\CLm}\times(\Vars_T\to
D)\times(\Vars_S\to S)\to D$ and
$\intP{-}^S:\Stream_{\CLm}\times(\Vars_T\to D)\times(\Vars_S\to S)\to S$
are defined by:
\begin{align*}
 \intP{C}^T_{\rho,\theta} & = c & \intP{\a}^S_{\rho,\theta} & = \theta(\a)\\
 \intP{x}^T_{\rho,\theta} & = \rho(x) & \intP{T::\St}^S_{\rho,\theta} & = \intP{T}^T_{\rho,\theta}::\intP{\St}^S_{\rho,\theta},\\
 \intP{T\cdot U}^T_{\rho,\theta} & = \intP{T}^T_{\rho,\theta}\cdot\intP{U}^T_{\rho,\theta}\\
 \intP{T\star \St}^T_{\rho,\theta} & = \intP{T}^T_{\rho,\theta}\star\intP{\St}^S_{\rho,\theta}
\end{align*}
where $c$ denotes the element of $D$ corresponding to the constant $C$,
that is, $\intP{\K_0}^T_{\rho,\theta}=\k_0$,
$\intP{\S_0}^T_{\rho,\theta}=\s_0$, and so on. We often omit the
 superscript $T$ or $S$.
\end{definition}

\begin{theorem}[Soundness and completeness]\label{thm:sca}
For any $\CLm$-terms $T$ and $U$, $T=_\CLm U$ iff
$\intP{T}_{\rho,\theta}=\intP{U}_{\rho,\theta}$ in any extensional
stream combinatory algebra for any $\rho$ and $\theta$.
\end{theorem}

\begin{proof}
 (Only-if part) The soundness can be proved by straightforward induction
 on $T=_\CLm U$.

 (If part) We can construct a term model as follows.
Let $D=\Term_\CLm/=_\CLm$ and $S=\Stream_\CLm/=_\CLm$, and the
equivalence classes in $D$ and $S$ are denoted such as $[T]$ and
$[\St]$. The operations are defined as $[T]\cdot[U]=[T\cdot U]$,
$[T]\star[\St]=[T\star\St]$, and $[T]::[\St]=[T::\St]$. The element
$\k_0$ is defined as $[\K_0]$ and similar for the other constants. The
resulting structure is easily proved to be an extensional stream
combinatory algebra. If we take $\rho$ and $\theta$ as $\rho(x)=[x]$ and
$\theta(\a)=[\a]$, respectively, then $\intP{T}_{\rho,\theta}=[T]$ for
any $T\in\Term_\CLm$, hence we have that
$\intP{T}_{\rho,\theta}=\intP{U}_{\rho,\theta}$ implies $T=_\CLm U$.
\end{proof}

\begin{corollary}
For any $\lm$-terms $M$ and $N$, $M=_{\lm} N$ iff
$\intP{M^*}_{\rho,\theta}=\intP{N^*}_{\rho,\theta}$ in any extensional
stream combinatory algebra for any $\rho$ and $\theta$.
\end{corollary}

\begin{proof}
 It immediately follows from Theorem \ref{thm:lm-and-scl} and Theorem
 \ref{thm:sca}.
\end{proof}

\section{Algebraic Characterization of Stream Models}

Definition \ref{def:ext} of the extensional stream models is a direct
one, but it depends on the definability of the meaning function on the
$\lm$-terms. In this section, we give a syntax-free characterization for
the extensional stream models, that is, the class of the extensional
stream models exactly coincides with the subclass of the extensional
stream combinatory algebras in which $S$ is a stream set on
$D$.

\begin{definition}\rm
A stream applicative structure $\tup{D,S,\cdot,\star,::}$ is {\em
 standard} if $\tup{S,::}$ is a stream set on $D$.
\end{definition}

Note that, for standard stream applicative structures, the
extensionality for term application $(\cdot)$ follows from the
extensionality for $(\star)$ since $(::)$ is surjective: suppose
$d_1\cdot d=d_2\cdot d$ for any $d\in D$, then for any $s\in S$
we have $d_1\cdot d \star s=d_2\cdot d\star s$, which means
$d_1\star(d::s)=d_2\star(d::s)$ for any $d$ and $s$. Hence $d_1=d_2$ by
the extensionality with respect to $\star$.

\begin{theorem}
For a non-empty set $D$ and a stream set $\tup{S,::}$ on $D$, the
following are equivalent.

1. $\tup{D,S}$ is an extensional stream model with some $[S\to D]$ and $\Psi$.

2. $\tup{D,S}$ is a standard extensional combinatory algebra with some
 operations $(\cdot)$ and $(\star)$, and some elements $\k_0$, $\k_1$,
 $\s_0$, $\s_1$, $\c_{10}$, $\c_{11}$, $\w_1$ in $D$.
\end{theorem}

\begin{proof}
(1$\Longrightarrow$2) Suppose
$\tup{D,S,[S\to D],::,\Psi}$ is an extensional stream model. Define
\begin{align*}
 d\star s & = \Phi(d)(s)
 & d\cdot d' & = \Psi(\olam s\in S.\Phi(d)(d'::s)),
\end{align*}
where we should note that $d \cdot d'$ is identical to
$\intp{xy}_{\rho[x\mapsto d,y\mapsto d']}$ and hence it is always
defined. Define $\k_0 = \intp{\lambda xy.x}$ and so on.
% \begin{align*}
%  \k_0 & = \intp{\lambda xy.x}
%  & \k_1 & = \intp{\lambda x.\mu\a.x}\\
%  \s_0 & = \intp{\lambda xyz.xz(yz)}
%  & \s_1 & = \intp{\lambda xy.\mu\a.x\a(y\a)}\\
%  \c_{10} & = \intp{\lambda x.\mu\a.\lambda y.xy\a}
%  & \c_{11} & = \intp{\lambda x.\mu\a\b.x\b\a}\\
%  \w_1 & = \intp{\lambda x.\mu\a.x\a\a},
% \end{align*}
% where all the $\lm$-terms in the right-hand sides contain no free
% variable, so they are independent of $\rho$ and $\theta$.
Then $\tup{D,S,\cdot,\star,::}$ is a standard extensional stream combinatory
algebra.  Indeed, it is a stream applicative structure, since
\begin{align*}
 d_1\cdot d_2 \star s_3 = \Phi(\Psi(\olam s.\Phi(d_1)(d_2::s)))(s_3)
  = \Phi(d_1)(d_2::s_3)
  = d_1\star(d_2::s_3).
\end{align*}

(2$\Longrightarrow$1) Suppose $\tup{D,S,\cdot,\star,::}$ is a standard
extensional stream combinatory algebra. Define $[S\to D]:=\{f_d
\mid d\in D\}$, where $f_d$ denotes $\olam s\in S.d\star s$. Then
$\Phi(d) = f_d$ and $\Psi(f_d) = d$
are well-defined since $D$ is extensional, and they give a bijection
 between $[S\to D]$ and $D$.
% If $f_d=f_{d'}$, we have $f_d(s)=f_{d'}(s)$ for any $s\in % S$, that
% is, $d\star s=d'\star s$. Hence we have $d=d'$ by the extensionality.
We can see that the interpretation $\intprt{M}$ with respect to $\Phi$ and
$\Psi$ coincides with $\intP{M^*}_{\rho,\theta}$. That is shown by the
following lemmas for any $\CLm$-term $T$:
\begin{align*}
\intP{\lambda^*x.T}_{\rho,\theta}\cdot d &= \intP{T}_{\rho[x\mapsto d],\theta}
&
\intP{\mu^*\a.T}_{\rho,\theta}\star s &= \intP{T}_{\rho,\theta[\a\mapsto s]}.
\end{align*}
In the case of $M=\lambda x.N$, $\intprt{M}=\intP{M^*}_{\rho,\theta}$
 is proved as follows.
\begin{align*}
\intP{M^*}_{\rho,\theta}\star(d::s)
&= \intP{M^*}_{\rho,\theta}\cdot d\star s\\
&= \intP{N^*}_{\rho[x\mapsto d],\theta}\star s & \mbox{(by the lemma)}\\
&= \intp{N}_{\rho[x\mapsto d],\theta}\star s & \mbox{(by IH)}\\
&= \Phi(\intp{N}_{\rho[x\mapsto d],\theta})(s)
\end{align*}
Therefore we have $\olam d::s.\Phi(\intp{N}_{\rho[x\mapsto
d],\theta})(s) = \olam d::s.\intP{M^*}_{\rho,\theta}\star(d::s) =
f_{\intP{M^*}_{\rho,\theta}} \in [S\to D]$, and hence
$\intp{M}_{\rho,\theta}$ is defined and identical to
$\Psi(f_{\intP{M^*}_{\rho,\theta}})=\intP{M^*}_{\rho,\theta}$. The other
cases are similarly proved.  Hence, $\tup{D,S,[S\to D],::,\Psi}$ is an
extensional stream model.
\end{proof}

\section{Conclusion}

We have proposed models of the untyped $\lm$-calculus:
the set-theoretic and the categorical stream models, and the stream combinatory algebras. We have also
shown that extensional stream models are algebraically characterized as
a particular class of the extensional stream combinatory algebras. The
stream combinatory algebra has been induced from the new combinatory
calculus $\CLm$, which exactly corresponds to the untyped
$\lm$-calculus.
%Terms of $\CLm$ can be
%seen as combinators for both data and streams.

% In this paper, we formalize $\CLm$ with the seven constants, but it just
% means that the seven constants are sufficient. (As noted below, the
% choice of the indexed constants seems suitable for the generalization to
% the stream hierarchy.) Remind that the combinators $\C$ and $\W$ in
% $\CL$ are definable from $\K$ and $\S$. It is an interesting question
% whether some of the constants of $\CLm$ are definable from the others.

\subsection{Related Work}
{\bf Models of the untyped $\lambda\mu$-calculus.} 
In \cite{Streicher-Reus1998}, Streicher and Reus proposed the
continuation models for the untyped $\lambda\mu$-calculus (which is a
variant of Parigot's original $\lambda\mu$-calculus) based on the
idea that the $\lambda \mu$-calculus is a calculus of
continuations. If we see each stream $d :: s$ as a pair $\left(
  d, s \right)$ of a function argument $d$ and a continuation $s$, the
meaning function for the stream models looks exactly the same as that
for the continuation models.

In the untyped $\lambda \mu$-calculus in
  \cite{Streicher-Reus1998}, the named terms are distinguished from
  the ordinary terms. In the continuation models, an object $R$
  (called {\em response object}) for the denotations of named terms is
  fixed first, then the object $D$ for the denotations of the ordinary
  terms and the object $S$ for continuations are respectively given as
  the solutions of the following simultaneous recursive equations:
\begin{equation}
  D \times S \cong S, \hspace{1em} S \Rightarrow R \cong D. \label{eq:reus}
\end{equation}
These equations say that the continuations are streams of ordinary
terms, and the ordinary terms can act as functions from continuations
to responses (i.e.  results of computations).  On the other hand, in
the $\Lambda\mu$-calculus, the named terms and terms are integrated
into one syntactic category, thus allowing us to pass terms to named
terms, such as $M \alpha N$.  In the model side, this extension
corresponds to that the response object $R$ in \eqref{eq:reus} is
replaced by $D$, resulting in the simultaneous recursive equations
\eqref{eq:streameq}.

% In
% \cite{Streicher-Reus1998}, Streicher and Reus proposed the continuation
% models for the untyped $\lambda\mu$-calculus, which is based on the idea
% that the $\lambda\mu$-calculus is a calculus of continuations. If we see
% each stream $d::s$ as a pair $\tup{d,s}$ of \EDIT{} a function argument
% $d$ and a continuation $s$, the meaning function for the stream models
% looks exactly the same as that for the continuation models. The main
% difference between them is that the continuation models are only for (a
% variant of) the Parigot's original $\lambda\mu$-calculus, in which the
% named terms are distinguished from the ordinary terms. In the
% continuation models, ordinary terms are interpreted as functions from
% continuations to responses (i.e. results of computations), whereas named
% terms are interpreted as responses. Hence, for example, a term of the
% form $(M)\a N$ cannot be interpreted in the continuation models. On the
% other hand, in the extensional stream models, both terms and named terms
% are uniformly interpreted in $D$. Moreover, the variant of
% $\lambda\mu$-calculus in \cite{Streicher-Reus1998} can be interpreted in
% the stream models. In the extended syntax, we have continuations such as
% $M::N::\a$, which is interpreted as $\intp{M}::\intp{N}::\theta(\a)\in
% S$. \COMMENT{Can we say something on the differences between
% $\lambda\mu$ and $\Lambda\mu$ from semantical point of view?}

In \cite{vanBakelEtAl2011}, van Bakel et al. considered
intersection type systems and filter models for the
$\lambda\mu$-calculus based on the idea of the continuation models of
Streicher and Reus. They considered only the original
$\lambda\mu$-calculus, and $\Lambda\mu$-terms such as $\mu\a.x$ have no
type except for $\omega$ in the proposed intersection type system, and
hence, they are interpreted as the bottom element in the filter
model. They also showed that every continuation model can be a model of
the $\lm$-calculus. The idea is to translate each $\lm$-term to a
$\lambda\mu$-term as $\mu\a.M$ to $\mu\a.M\a$ and $M\a$ to $\mu\b.M\a$
with a fresh $\b$. However, as pointed out in \cite{vanBakelEtAl2011},
the axiom ($\beta_S$) is unsound for this interpretation in general,
whereas it is sound in our stream models.

Akama \cite{Akama2001} showed that the untyped $\lambda\mu$-calculus can
be interpreted in partial combinatory algebras. It is based on the idea
that $\mu$-abstractions are functions on streams. However, it restricts
terms to affine ones, that is, each bound variable must not occur more
than once.

Fujita \cite{Fujita2002} considered a reduction system for the
$\lambda\mu$-calculus with ($\beta_T$), ($\eta_T$), ($\mu$), and (${\it
fst}$) rules, and gives a translation from the $\lambda\mu$-calculus to
the $\lambda$-calculus which preserves the equality, and hence it is
shown that any extensional $\lambda$-model is a model of the
$\lambda\mu$-calculus.  In the translation, each $\mu$-abstraction is
interpreted as a potentially infinite $\lambda$-abstraction by means of
a fixed point operator. However, it considers neither ($\beta_S$) nor
($\eta_S$), and it seems hard to obtain a similar result for them.
% Every extensional stream model is an extensional $\lambda$-model, but it
% should be further studied how the extensional stream models relate to
% the $\lambda$-models: Is there any extensional $\lambda$-model which
% cannot be an extensional stream model?
%\COMMENT{Deeper comparison is required(?).}

{\bf Combinatory logic and classical logic.}
Baba et al. considered some extensions of the $\lambda$-calculus with
combinators corresponding to classical axioms such as Peirce's law and
double negation elimination in \cite{BKH2000}.

Nour \cite{Nour2006} introduced the classical combinatory logic
corresponding to Barbanera and Berardi's symmetric $\lambda$-calculus
\cite{BB1994}. The classical combinatory logic has two kinds of
application operators: one is the ordinary function application, and the
other represents the interaction of terms and continuations, which is
based on the same idea as the stream application operator in $\CLm$ (and
denoted by the same symbol $\star$). Nour's classical combinatory logic
is a typed calculus corresponding to classical logic, and its weak
reduction corresponds to the reduction of the symmetric
$\lambda$-calculus. On the other hand, we have not found any reasonable
type system for $\CLm$ as discussed below, but $\CLm$ corresponds to the
$\lm$-calculus, and, in particular, it can represent the
$\mu$-abstraction over continuation variables.
% Moreover, we can extend
% the strong reduction of $\CL$ to $\CLm$, the confluence of which follows
% from the confluence of the $\lm$-calculus.

\subsection{Further Study}
{\bf (Extensional) stream models.}
% We have not shown any concrete
% extensional stream model in this paper, but we can obtain (non-trivial)
% models as solutions of the domain equation $D\simeq D^\omega\to D$ in an
% appropriate category such as ${\rm CPO^E}$ of embeddings on CPOs. Let
% $B$ be a non-trivial domain, $C$ be $B^\omega$, and $D$ be a solution of
% $D\simeq D\to C$. Then we have $C^\omega \simeq B^{\omega\times\omega}
% \simeq B^\omega =C$, and therefore $D \simeq D\to C \simeq D\to C^\omega
% \simeq (D\to C)^\omega \simeq D^\omega$. Hence, we have $D \simeq
% D^\omega \to C \simeq D^{\omega + 1}\to C \simeq D^\omega \to(D\to C)
% \simeq D^\omega\to D$.
One natural direction of study is to analyze the local structure of the
domain-theoretic extensional stream models constructed from the
solutions of \eqref{eq:nakazawa} in Section \ref{sec:categorical}. How
do they relate to the B\"{o}hm-tree representation proposed in
\cite{Saurin2010b}?  Do these models enjoy the approximation theorem?
Which syntactic equality corresponds to the equality in these models?

% We can also generalize the class of stream models to non-extensional
% ones. The definition of the {\it stream models} is the same as the
% extensional stream models except for the condition
% $\Psi\circ\Phi=\id$. In such stream models, however, the axiom
% ($\beta_T$) is preserved only for the restricted form:
% \[
%  (\lambda x.M)N_1\cdots N_n\a =_{\beta_T} M[x:=N_1]\cdots N_n\a.
% \]
% Hence, if we restrict the $\lambda\mu$-terms to
% \[
% M ::= x \mid \lambda x.M \mid \mu\a.M \mid (M)M_1\cdots M_n\a,
% \]
% the calculus is sound with respect to the stream models.

We have considered only extensional theories and models in this
paper. We can na\"{i}vely define non-extensional stream models by
weakening the condition $[S\to D]\simeq D$ to $[S\to D]\lhd D$, and then
the functions $\Phi_0$ and $\Psi_0$ in Theorem \ref{thm:lambda-model}
are still well-defined.  However, under such a structure, we always have
$\Psi_0\circ \Phi_0 = {\rm id}$, so the extensionality axiom $\eta_T$ is
unexpectedly sound, for example $\intp{\lambda xy.xy}=\intp{\lambda
x.x}$ always holds.
% Intuitively, the reason is
% $[S\to D]\simeq [D\to S\to D]\simeq [D\to D\to S\to D]$ holds due to
% $S\simeq D\times S$ even in the ``non-extensional'' models. 
Furthermore, we do not know how to derive that $\Phi_0\circ\Psi_0={\rm
id}$, which is essential for modeling the $\beta$-equality of the term
application. It is future work to study how we can define appropriate
notion of the models of the non-extensional $\lm$-calculus.

Moreover, syntactic correspondence between non-extensional
theories of the $\lm$-calculus and $\CLm$ is still unclear and it is
future work to study on it.

{\bf Types and classical logic.} The $\lambda\mu$-calculus was
originally introduced as a typed calculus corresponding to the classical
natural deduction in the sense of the Curry-Howard isomorphism. It is
future work to adapt our discussion to a typed setting and to
study the relationship to classical logic. It is
well-known that the combinatory logic with types exactly corresponds to
the Hilbert-style proof system of intuitionistic logic. On the other
hand, it is unclear how we can consider $\CLm$ as a typed calculus,
since the $\lm$-terms corresponding to the constants of $\CLm$
are not typable in the ordinary typed $\lambda\mu$-calculus, for
example, $(\S_1)_*=\lambda xy.\mu\a.x\a(y\a)$.

\paragraph*{}
{\bf Acknowledgments} We are grateful to Dana Scott, Kazushige Terui,
Makoto Tatsuta, and anonymous reviewers for helpful comments, and to
Daisuke Kimura for fruitful discussions.

\paragraph*{}
\bibliographystyle{plain}
\bibliography{main}

%\appendix

\end{document}